\documentclass[11pt]{article}
\usepackage{graphicx}

\usepackage[centertags]{amsmath}
\usepackage{amsfonts}
\usepackage{color}
\usepackage{amssymb}
\usepackage{amsthm}
\usepackage{newlfont}
    \usepackage[ansinew]{inputenc}

\newtheorem{theorem}{Theorem}

\newtheorem{lemma}[theorem]{Lemma}

\renewenvironment{proof}[1][]%
 {\noindent {\setcounter{claim}{0}\bf Proof ---
   }{#1}{}}{\hfill$\Box$\vspace{2ex}}

\def\Box{\hbox{\vrule width6pt height6pt depth0pt}}

{\newdimen\quadamount\quadamount=1.5em%
\def\TAB##1{\\\hskip##1\quadamount\relax}%
\medskip\sf\begin{tabular}{l}}%
{\end{tabular}\medskip}%

\begin{document}


\title{Edge-colouring and total-colouring chordless graphs}

\author{Raphael C. S. Machado\thanks{{Inmetro, Rio de Janeiro,
      Brazil.} {Partially supported by Brazilian agencies CNPq and
      Faperj.} rcmachado@inmetro.gov.br}, 
Celina M. H. de Figueiredo\thanks{{COPPE, Universidade Federal do Rio
    de Janeiro, Brazil.} {Partially supported by Brazilian agencies
    CNPq and Faperj.} celina@cos.ufrj.br}, 
Nicolas~Trotignon\thanks{{CNRS, LIP, ENS Lyon, France.} {Partially supported by \emph{Agence Nationale de la Recherche} under reference
    \textsc{anr 10 jcjc 0204 01} and by PHC Pavle Savi\'c grant,
    jointly awarded by EGIDE, an agency of the French Minist\`ere des
    Affaires \'etrang\`eres et europ\'eennes, and Serbian Ministry for
    Science and Technological Development. Also: INRIA, Universit\'e
    de Lyon.} nicolas.trotignon@ens-lyon.fr}
}

\date{May 14, 2013}
\maketitle

\begin{abstract}
  A graph $G$ is \emph{chordless} if no cycle in $G$ has a chord.  In
  the present work we investigate the chromatic index and total
  chromatic number of chordless graphs.  We describe a known
  decomposition result for chordless graphs and use it to establish that
  every chordless graph of maximum degree $\Delta\geq 3$ has chromatic
  index $\Delta$ and total chromatic number $\Delta + 1$.  The proofs are
  algorithmic in the sense that we actually output an optimal colouring of a
  graph instance in polynomial time. 

\end{abstract}


\def\baselinestretch{1}
\typeout{Introduction}

\section{Introduction}
\label{s:introduction}

Let $G=(V,E)$ be a simple graph. The degree of a vertex $v$ in $G$ is
denoted by $\deg_{G}(v)$ while the maximum degree in $G$ is
denoted by $\Delta(G)$.  

\subsection*{Edge-colouring}

An \emph{edge-colouring} of $G$ is a function
$\pi:E\rightarrow\mathbf{C}$ such that no two adjacent edges receive
the same colour $c\in\mathbf{C}$. If $\mathbf{C}=\{1,2,...,k\}$, we
say that $\pi$ is a $k$-\emph{edge-colouring}.  The \emph{chromatic
index} of $G$
is the least~$k$ for which $G$
has a $k$-edge-colouring.
The chromatic index of any graph $G$ is obviously at least
$\Delta(G)$.  
Vizing's theorem~\cite{Vizing} states that every graph $G$,
is $(\Delta(G)+1)$-edge-colourable.

It is NP-complete to determine whether a graph $G$ is $\Delta(G)$-edge-colourable~\cite{Holyer,Leven}.
The edge-colouring problem remains NP-complete for perfect graphs~\cite{CaiEllis}. 
%
Graph classes for which edge-colouring is polynomially solvable include
bipartite graphs~\cite{Johnson}, split-indifference graphs~\cite{SplitIndifference},
series-parallel graphs (hence outerplanar)~\cite{Johnson}, and
$k$-outerplanar graphs, for $k\geq1$~\cite{Bodlaender}.
The complexity of edge-colouring is unknown for several well-studied
strongly structured graph classes, for which only partial results have
been reported, such as cographs~\cite{Barbosa}, join
graphs~\cite{join1,JoinCobipartite}, planar
graphs~\cite{Sanders}, chordal graphs, and several
subclasses of chordal graphs such as indifference
graphs~\cite{Celina}, split graphs~\cite{ChenFuKo} and interval
graphs~\cite{Bojarshinov}.

%

\subsection*{Total-colouring}

An \emph{element} of a graph is either a vertex or an edge.  Two
elements are \emph{adjacent} if they are either adjacent vertices,
adjacent edges or a vertex incident to an edge.  A \emph{total-colouring} of $G$ is a function $\pi:V\cup E\rightarrow\mathbf{C}$
such that no two adjacent elements receive the same colour
$c\in\mathbf{C}$. If $\mathbf{C}=\{1,2,...,k\}$, we say that $\pi$ is
a $k$-\emph{total-colouring}.  The \emph{total chromatic number} of
$G$
is the least~$k$ for which $G$ has a
$k$-total-colouring.
The total chromatic number of any graph $G$ is obviously at least
$\Delta(G)+1$.  
The \emph{Total Colouring Conjecture (TCC)},
posed independently by Behzad~\cite{BehzadTese} and Vizing~\cite{Vizing},
states that every simple graph $G$ 
would be 
$(\Delta(G)+2)$-total-colourable
and it is a challenging open problem in Graph Theory.

It is NP-complete to determine whether a graph $G$ is $(\Delta(G)+1)$-total-colourable~\cite{SanchezArroyo1}.
The total-colouring problem remains NP-complete when restricted to 
$r$-regular bipartite inputs~\cite{SanchezArroyo}, for each fixed $r\geq3$.
The total-colouring 
problem is known to be polynomial --- and the TCC is valid --- for few 
very restricted graph classes, such as cycles, complete graphs, 
complete bipartite graphs~\cite{Yap}, grids~\cite{ChristianeBipartidos}, 
series-parallel graphs~\cite{SP-choose-Delta4,WangPang,SP-choose-Delta3}, 
and split-indifference graphs~\cite{ChristianeSplitIndifference}.
%
%
%
%
%
The computational complexity of the total-colouring problem is unknown for 
several important and well-studied graph classes. The complexity of 
total-colouring planar graphs is unknown; in fact, even the TCC 
has not yet been settled for planar graphs~\cite{WeifanWang}. The complexity of 
total-colouring is open for the class of chordal graphs, and the partial results 
for the related classes of interval graphs~\cite{Bojarshinov}, 
split graphs~\cite{ChenFuKo} and dually chordal graphs~\cite{Dually} expose 
the interest in the total-colouring problem restricted to chordal graphs. Another 
class for which the complexity of total-colouring is unknown is the class of 
join graphs: the results found in the literature consider very restricted 
subclasses of join graphs~\cite{Hilton,OneKindJoin}.

\subsection*{Chordless graphs}

A \emph{cycle} $C$ in a graph $G$ is a sequence of vertices
$v_1v_2\ldots v_nv_1$, that are distinct except for the first and the
last vertex, such that for $i=1, \ldots ,n-1$, $v_iv_{i+1}$ is an edge
and $v_nv_1$ is an edge --- we call these edges \emph{the edges of~
  $C$}.  An edge of $G$ with both endvertices in a cycle $C$ is called
a \emph{chord} of~$C$ if it is not an edge of $C$.  A \emph{chordless
  graph} is a graph whose cycles are all chordless.  Chordless graphs
were first studied independently by Dirac~\cite{dirac:chordless} and
Plummer~\cite{plummer:68}.  In particular, they give several
characterizations of these graphs, prove that they are
3-vertex-colourable and that they are 2-connected if and only if they
are minimally 2-connected (meaning that the removal of any edge yields
a non 2-connected graph).  L{\'e}v{\^e}que, Maffray and
Trotignon~\cite{ISK4} studied the class of chordless graphs independently.
The motivation was to compute the
chromatic number of line graphs with no induced generalized wheels, where a
\emph{generalized wheel} is a graph made of a cycle together with a vertex that
has at least 3 neighbours on the cycle.  They observed that these line
graphs are the line graphs of chordless graphs of maximum degree at
most 3, and they prove that chordless graphs of maximum degree at most
3 are 3-edge-colourable.  On the way to this result, they prove a
decomposition theorem for chordless graphs that is seemingly
independent of the results in~\cite{dirac:chordless,plummer:68}.

There is a second motivation for the study of the chromatic index of
chordless graphs: they are a subclass of the class of
\emph{unichord-free} graphs which are graphs that do not contain, as
induced subgraph, a cycle with unique chord.  Trotignon and Vu\v
skovi\'c ~\cite{tv} proved a decomposition theorem for these graphs,
that in fact contains implicitly the decomposition of chordless graphs
from~\cite{ISK4}.  Both edge-colouring and total-colouring problems are 
NP-complete problems when restricted to unichord-free graphs, 
as proved by Machado, de Figueiredo and Vu{\v s}kovi{\'c}~\cite{Edge-tv}
and by Machado and de Figueiredo~\cite{Total-tv}; hence, it is of interest to
determine subclasses of unichord-free graphs for which edge-colouring
and total-colouring are polynomial.

Our main result is the following.

\begin{theorem}
  \label{th:main}
  Let $G$ be a chordless graph of maximum degree at least~3.  Then $G$
  is $\Delta(G)$-edge-colourable and $(\Delta(G)+1)$-total-colourable.
  Moreover, these colourings may be obtained in time $O(|V(G)|^3|E(G)|)$.
\end{theorem}

Note that the edge-colouring of chordless graphs with maximum degree~3
was already established in \cite{ISK4}. 
Theorem~\ref{th:main} relies on the decomposition theorem
from~\cite{ISK4}.  
We emphasize, however, that there are differences
between the proof of~\cite{ISK4} and ours. Most remarkably, since~\cite{ISK4} 
deals only with the case $\Delta=3$, any cutset of two non-adjacent vertices 
actually determines a cutset of two non-adjacent edges. 
Such an edge-cutset is used to construct a natural induction on the decomposition blocks.
Our proof uses a different strategy based on the existence of an extreme decomposition tree, 
in which one of the decomposition blocks is 2-sparse. This leads to our third and main motivation
for this work: to understand how such kind of decomposition results, which are classically applied
to the design of vertex colouring algorithms, can be useful in the development of algorithms for
other colouring problems, in particular edge-colouring~\cite{Edge-tv} and total-colouring~\cite{Total-tv,Total-tv-square} --- the present work is successful in the sense that our chosen class of chordless graphs showed to be fruitful
for the development of polynomial-time edge-colouring and total-colouring algorithms.

Section~\ref{s:structure} reviews the decomposition result for
chordless graphs established in~\cite{ISK4}. Section~\ref{s:2-sparse} gives several results for a
subclass of chordless graphs, the so-called \emph{2-sparse graphs}.
Section~\ref{s:proof} gives the proof of Theorem~\ref{th:main}.

\section{Structure of chordless graphs}
\label{s:structure}

The goal of the present section is to review the structure theorem
from~\cite{ISK4}.  A graph is \emph{2-sparse} if every edge is incident
to at least one vertex of degree at most~2.  A 2-sparse graph is
chordless because any chord of a cycle is an edge between two vertices
of degree at least three.  

A \emph {proper 2-cutset} of a connected graph $G=(V,E)$ is a pair of
non-adjacent vertices $a, b$ such that $V$ can be partitioned into
non-empty sets $X$, $Y$ and $\{ a,b \}$ so that: there is no edge between $X$ and $Y$; and both $G[X \cup \{
a,b \}]$ and $G[Y \cup \{ a,b \}]$ contain an $ab$-path and neither of 
$G[X \cup \{a,b \}]$ nor $G[Y \cup \{ a,b \}]$ is a chordless path.
We say that $(X, Y, a, b)$ is a \emph{split} of
this proper 2-cutset.

  The decomposition result stated in Theorem~\ref{t:chordless} is implicit in \cite{tv}, explicit in~\cite{ISK4}, and
  for the sake of completeness, we include a short proof
  from~\cite{aboulkerRTV:propeller}.

\begin{theorem}[L{\'e}v{\^e}que, Maffray and Trotignon~\cite{ISK4}]
  \label{t:chordless}
  If $G$ is a 2-connected chordless graph, then either $G$ is 2-sparse
  or $G$ admits a proper 2-cutset.
\end{theorem}

\begin{proof}
  If $G$ is not 2-sparse, then it has an edge $e=uv$ such that $u$ and
  $v$ have both degree at least~3.  If $G\setminus e$ is 2-connected,
  then it contains a cycle through $u$ and $v$, hence $uv$ is a chord
  in $G$, a contradiction.  It follows that $G \setminus e$ is
  disconnected or has a cutvertex.

  If $G \setminus e$ is disconnected, then $u$ and $v$ are
  cutvertices 
  of $G$, a contradiction. So $G\setminus e$ is connected and has
  a cutvertex $w$.  Since $G$
  is 2-connected, the graph $(G\setminus e) \setminus w$ has exactly
  two connected components $G_u$ and $G_v$, containing $u$ and $v$
  respectively, and $V = V(G_u) \cup V(G_v) \cup \{w\}$.  Let $u'\notin \{v,
  w\}$ be a neighbour of $u$ ($u'$ exists since $u$ has degree at least
  3).  So, $u'\in V(G_u)$.  In $G$, vertex $u$ is not a cutvertex, so there is a
  path $P_u$ from $u'$ to $w$ whose interior is in $G_u$.  Together
  with a path $P_v$ from $v$ to $w$ with interior in $G_v$, $P_u$,
  $uu'$ and $e$ form a cycle, so $uw\notin E(G)$ for otherwise $uw$
  would be a chord of this cycle. 
  Also, since $u$ is of degree at least 3 and $G$ is chordless, $u$ has a neighbour in $G_u\setminus P_u$.
  Hence $(V(G_u)\setminus \{u\}, V(G_v), u, w)$ is a split of a proper 2-cutset of $G$.
 \end{proof}

 The \emph{blocks} $G_X$ and $G_Y$ of a graph $G$ with respect to a
 proper 2-cutset with split $(X, Y, a, b)$ are defined as
 follows. Block $G_X$ (resp.\ $G_Y$) is the graph obtained by taking
 $G[X\cup\{a,b\}]$ (resp.\ $G[Y\cup\{a,b\}]$) and adding a new vertex
 $w$ adjacent to $a, b$.  Vertex $w$ is a called the \emph{marker
 vertex} of the block $G_X$ (resp.\ $G_Y$). 
 Clearly, blocks $G_X$ and $G_Y$ of a 2-connected chordless graph
 are 2-connected chordless graphs as well.
%
%
%

Theorem~\ref{t:extremal} states that a not 2-sparse chordless graph has
an extremal decomposition, that is, a decomposition in which (at least) 
one of the blocks is 2-sparse.

\begin{theorem}
\label{t:extremal}
Let $G$ be a 2-connected not 2-sparse chordless graph.  Let $(X, Y, a,
b)$ be a split of a proper 2-cutset of $G$ such that $|X|$ is minimum among
all possible such splits.  Then $a$ and $b$ both have
at least two neighbours in $X$, and $G_{X}$ is 2-sparse.
\end{theorem}

\begin{proof}
  First, we show that $a$ and $b$ both have  at least two neighbours in $X$. 
  Suppose that one of $a,b$, say $a$,  has a
  unique neighbour $a' \in X$.  We claim that $a'$ is not adjacent to
  $b$.  For otherwise, since $G[X \cup \{a, b\}]$ does not induce a
  path and $G$ is 2-connected, there is a path in $G[X\cup\{b\}]\setminus a'b$ from $b$ to
  $a'$, that 
  together with the edge $a'a$ and a path
  from $a$ to $b$ with interior in $Y$
  forms a cycle with a chord: $a'b$.  So, $a'$ is not adjacent
  to $b$.  Hence, by replacing $a$ by $a'$, we obtain a proper
  2-cutset that contradicts the minimality of $X$. 

  Let $w$ be the marker vertex of $G_{X}$.
  Suppose that 2-connected chordless graph $G_X$ is not 2-sparse. According to
  Theorem~\ref{t:chordless}, $G_X$ has a proper 2-cutset with split
  $(X_1,X_2,u,v)$.  Choose it so that $u$ and $v$ both have degree at
  least 3 (this is possible as explained at the beginning of the
  proof).  Note that $w\notin \{u, v\}$. W.l.o.g. $w\in X_1$.
  But then $\{a,b\}\subseteq X_1\cup\{u,v\}$ and hence $(X_2,Y\cup X_1\setminus\{w\},u,v)$ 
  is a proper 2-cutset of $G$, contradicting the
  minimality of $|X|$. Therefore, $G_X$ is 2-sparse.

%
%
%
\end{proof}

\section{2-Sparse graphs}
\label{s:2-sparse}

A 2-sparse graph is easily shown to be $\Delta(G)$-edge-colourable and $(\Delta(G)+1)$-total-colourable  (see below). 
But we need slightly more for the induction in our main proof.

\begin{theorem}[K\" onig]
  \label{th:konig}
  Any bipartite graph $G$ is $\Delta(G)$-edge-colourable. 
\end{theorem}

\begin{theorem}[Borodin, Kostochka, and Woodall~\cite{Borodin}]
\label{t:BKW}
Let $G=(V,E)$ be a bipartite graph and suppose that a list $L_{uv}$ of colours
is associated to each edge $uv\in E$. If for each edge
$uv\in E$, $|L_{uv}|\geq\max\{\deg_{G}(u),\deg_{G}(v)\}$, then there
is an edge-colouring $\pi$ of $G$ such that, for each edge $uv\in E$,
$\pi(uv)\in L_{uv}$.
\end{theorem}

The proof of the preceding theorem in~\cite{Borodin} is quite
involved.  It relies on counting arguments in such a way that it does
not imply clearly an algorithm outputting the colouring whose existence
is proved.  We propose the following simple weakening stated in Lemma~\ref{t:list-edge-colouring-bipartite},
which is enough for our purpose and whose proof
is clearly algorithmic.

Let $V_{\geq3}(G)$ be the set of vertices with degree at least 3 in G. 
Note that if G is 2-sparse then $V_{\geq3}(G)$ is a stable (i.e., independent, possibly empty) set. 
If $S$ is a stable set in $V(G)$, let $B(G,S)$ denote the bipartite subgraph of $G$ 
whose vertices are the vertices of $S$ and their neighbours, and whose edges 
are the edges of $G$ that have one end in $S$.

\begin{lemma}
  \label{t:list-edge-colouring-bipartite}
  Let $G=(V,E)$ be a bipartite graph with a bipartition $(X, Y)$ such
  that $V_{\geq3}(G)\subset X$.  Suppose that a
  list $L_{uv}$ of colours is associated to each edge $uv\in E$ such
  that for all $u\in X$, all edges incident to $u$ receive the same
  list. If for each edge $uv\in E$, $|L_{uv}| \geq \max\{\deg_{G}(u),
  \deg_{G}(v)\}$, then there is an edge-colouring $\pi$ of $G$ such
  that, for each edge $uv\in E$, $\pi(uv)\in L_{uv}$.
\end{lemma}

\begin{proof}
  We prove the theorem by induction on the number of edges.  We may
  assume that $G$ is connected for otherwise we work on
  components separately.  If all edges receive the same list (in
  particular when there are no edges at all) then the result follows
  from Theorem~\ref{th:konig}.  So, let $xy$ and $x'y$ be adjacent
  edges with different lists.  Note that from our assumptions, $x, x'
  \in X$ and $y \in Y$.  W.l.o.g.\ colour~$1$ is available for $xy$
  and not for $x'y$.  Let $P = v_1, \dots, v_k$ be a sequence of
  vertices such that $v_1=y$, $v_2 = x$, vertices $v_1, \dots,
  v_{k-1}$ are distinct, and for all $1 \leq i < k$, $v_i v_{i+1} \in
  E(G)$ and $\deg_G(v_i) =2$.  Suppose that $P$ is maximal w.r.t.\
  these properties.  Note that $v_k$ has degree~2 if and only if $G$
  is a cycle.

  We colour $v_1v_2$ with Colour~1 and colour greedily the edges
  $v_2 v_3$, \dots, $v_{k-1}v_k$ with some available colour, that is
  for each edge a colour in the list of the edge, not used for the
  preceding edge.  Note that this gives a colouring of the edges of
  $P$, even if $v_{k-1}v_k = x'y$ (this happens when $G$ is a
  cycle), because in this case, the colour used for $v_1 v_2
  = xy$ is not in the list of $v_{k-1}v_k = x'y$.  We delete all the
  edges of $P$ and, if $v_k$ has degree at least~3, we remove from all
  lists of edges incident to it, the colour used for $v_{k-1}v_k$. We
  edge-colour what remains by the induction hypothesis.
\end{proof}

Lemma~\ref{t:list-edge-colouring-bipartite} 
extends to 2-sparse graphs in the following way. 

\begin{lemma}
  \label{t:list-edge-colouring-2-sparse}
  Let $G=(V,E)$ be a 2-sparse graph and suppose that a list $L_{uv}$ of colours
  is associated to each edge $uv\in E$.  Let $S$ be a stable
  set of $G$ such that $V_{\geq3}(G)\subset S$.  Suppose that for every vertex $u\in S$,
  all edges incident to $u$ receive the same list.  If for each edge
  $uv\in E$, $|L_{uv}|\geq\max\{ \deg_{G}(u),\deg_{G}(v)\}$ and for
  each edge $uv\in E$ with no end in $S$, $|L_{uv}|\geq 3$, then there
  is an edge-colouring $\pi$ of $G$ such that, for each edge $uv\in
  E$, $\pi(uv)\in L_{uv}$.
\end{lemma}

\begin{proof}
  Let $G'=B(G,S)$. Note that $G'\subset G$, so each edge of $G'$ has a list of colours.
  We edge-colour $G'$ by Lemma~\ref{t:list-edge-colouring-bipartite}.  It
  remains to colour the edges in $E(G) \setminus E(G')$.  Each of them
  has a list of colours of size 3 and is adjacent to at most two edges
  of~$G$.  So, they can be coloured greedily.
\end{proof}

\begin{lemma}
  \label{l:edgeColorSparse}
  A 2-sparse graph $G$ of maximum degree $\Delta\geq 3$ is
  $\Delta(G)$-edge-colourable.
\end{lemma}

\begin{proof}
  We associate to each edge the list $\{1, \dots,
  \Delta(G)\}$ of colours and apply Lemma~\ref{t:list-edge-colouring-2-sparse}.
\end{proof}

From here on, we study total colouring of 2-sparse graphs.

\begin{lemma}
  \label{l:extSparseDelta4}
  Let $G=(V, E)$ be a 2-sparse graph of maximum degree $\Delta \geq 4$
  and let $S$ be a stable set of $G$ such that $V_{\geq3}(G)\subset S$.  Suppose that
  each vertex of $S$ is precoloured with some colour from $\{1, \dots,
  \Delta + 1\}$.  Suppose that a list $L_{uv}$ of colours from
  $\{1,2,...,\Delta+1\}$ is associated to every edge $uv$ with one end
  in $S$, and for every vertex $u\in S$, all edges incident to $u$
  receive the same list.  Suppose that for each edge $uv$ with one end
  $u\in S$, $|L_{uv}|\geq\max\{\deg_{G}(u),\deg_{G}(v)\}$ and
  $u$ is not precoloured with a colour from $L_{uv}$.

  Then there is a total-colouring $\pi$ of $G$ with colours in $\{1,
  \dots \Delta +1\}$ that extends the precolouring of the vertices of $S$ and
  such that $\pi(uv)\in L_{uv}$ for each edge $uv$ with one end in $S$.
\end{lemma}

\begin{proof}
  Let $G'=B(G,S)\subset G$. We edge-colour $G'$ by
  Lemma~\ref{t:list-edge-colouring-bipartite}.  It remains to colour
  edges in $E(G) \setminus E(G')$ and the vertices of degree at most~2
  of $G$.  Each of them is an element of $G$ incident to at most four
  elements of $G$ and at least five colours are available.  So, they
  can be coloured greedily.
\end{proof}

The following shows how a total precolouring of the ends of a path can be
extended to the path when only four colours are available. 

\begin{lemma}
  \label{l:extPath}
  Let $k\geq 3$ be an integer and $P= p_1 \dots p_k$ a path.  Suppose
  that $p_1$, $p_1p_2$, $p_{k-1}p_k$ and $p_k$ are total-precoloured
  in such a way that $p_1$ and $p_k$ receive the same colour.  This
  can be extended to a 4-total-colouring of $P$.
\end{lemma}

\begin{proof}
  W.l.o.g.\ we suppose that $p_1$ and $p_k$ are precoloured with
  Colour~1, $p_1p_2$ is precoloured with Colour~2 and $p_{k-1}p_{k}$
  is precoloured with Colour~2 or~3.  
  We proceed by induction on $k$. W.l.o.g. we suppose that $p_1$, $p_1p_2$, $p_{k-1}p_k$ and $p_k$ 
  are precoloured with colours 1, 2, 2, 1 or 1, 2, 3, 1. If $k = 3$ then $p_2$ is the only uncoloured element; 
  give it Colour 4. If $k = 4$ then colour $p_2$, $p_2p_3$ and $p_3$ with colours 3,1,4. If $k\geq5$ 
  then colour $p_2$, $p_2p_3$, $p_{k-2}p_{k-1}$ and $p_{k-1}$ with colours 4, 3, 1, 4; this can be extended 
  to a total-colouring of $P$ by the induction hypothesis applied to path $p_2 . . . p_{k-1}$.
%
\end{proof}

The following shows how to extend a total precolouring in a 2-sparse
graph of maximum degree~3.

\begin{lemma}
  \label{l:extSparse}
  Let $G$ be a 2-connected 2-sparse graph of maximum degree~3.  Then $G$
  is 4-total-colourable.  Moreover suppose that $u$ is a vertex of
  degree~2 that has two neighbours $a, b$ of degree~3 and suppose that
  $a$, $b$, $au$, $ub$ receive respectively Colours 1, 1, 2, 3.  This
  can be extended to a total-colouring of~$G$ using $4$ colours.
\end{lemma}

\begin{proof}
  Let $G' = B(G,S) \subset G$, where $S = V_{\geq3}(G)$. 
  We first total-colour $G'$. We give Colour~1 to all vertices of $S$.
  Then by Theorem~\ref{th:konig} we edge-colour $G'$ with Colours $2,
  3, 4$ (up to a relabeling, it is possible to give Colour $2, 3$ to
  $au$, $ub$ respectively).  We extend this to a total-colouring of
  $G$ by considering one by one the paths of $G$ whose interior
  vertices have degree~2 and whose ends have both degree~3.  Note that
  since $G$ is 2-connected with maximum degree~3, these paths
  edge-wise partition~$G$ and vertex-wise cover~$G$.  Let $P =
  p_1\dots p_k$ ($k\geq 3$ since $G$ is 2-sparse) be such a path.  The
  following elements are precoloured: $p_1, p_1p_2, p_{k-1}p_k, p_k$.
  The precolouring satisfies the requirement of Lemma~\ref{l:extPath},
  so we can extend it to~$P$.
\end{proof}

\begin{lemma}
  \label{l:totalColorSparse}
  A 2-sparse graph $G$ of maximum degree $\Delta\geq3$ is
  $(\Delta+1)$-total-colourable.
\end{lemma}

\begin{proof}
  We may assume that $G$ is connected, since otherwise it suffices to
  total-colour its components.
  
  When $\Delta(G) = 3$, we prove by induction a formally stronger
  statement: if $\Delta(G) \leq 3$ then $G$ is 4-total-colourable.  If
  $G$ has at most two vertices then the result clearly holds.  If $G$
  is 2-connected and $\Delta(G) = 2$ then $G$ is a cycle and a
  4-total-colouring exists.  If $\Delta(G) = 3$ and $G$ is
  2-connected, then we may rely on Lemma~\ref{l:extSparse}.  So we may
  assume that $G$ is not 2-connected and has a cutvertex $v$.  So, let
  $X$ and $Y$ be two disjoint non-empty sets that partition
  $V\setminus \{v\}$, with no edges between them.  The graphs $G[X\cup
  \{v\}]$ and $G[Y\cup\{v\}]$ are 4-total-colourable by induction,
  and we assume up to a relabeling that $v$ has the same colour in
  both total-colourings and that the edges incident to $v$ are
  coloured in both with (at most 3) different colours.  So, we obtain a
  4-total-colouring of $G$ by giving to each element the colour it has
  in one of the blocks.

  When $\Delta(G) \geq 4$, we give Colour $1$ to each vertex of degree
  at least 3.  We associate to each edge $uv\in E(G)$ such that
  $\deg_G(u) \geq 3$ or $\deg_G(v) \geq 3$ the list $L_{uv} = \{2, \dots
  \Delta(G) +1\}$.  Then, we apply Lemma~\ref{l:extSparseDelta4} to
  find the colouring.
\end{proof}

\section{Proof of Theorem~\ref{th:main}}
\label{s:proof}

\begin{proof}
Let $G = (V, E)$ be a chordless graph of maximum degree $\Delta \geq
3$.  We shall prove that $G$ is $\Delta$-edge-colourable and
$(\Delta+1)$-total-colourable by induction on $|V|$.  By
Lemmas~\ref{l:edgeColorSparse} and~\ref{l:totalColorSparse}, this
holds for 2-sparse graphs (in particular for the claw, the smallest
chordless graph of maximum degree at least 3).

If $G$ is not 2-connected, an edge- (resp.\ total-) colouring of $G$
can be recovered from an edge (resp.\ total-) colouring of its blocks.
So, we may assume that $G$ is 2-connected but not 2-sparse. According to
Theorem~\ref{t:chordless}, $G$ has a proper 2-cutset with split $(X,
Y, a, b)$ and we choose such a 2-cutset with minimum
$|X|$.  By Theorem~\ref{t:extremal}, the block $G_X$
is 2-sparse.

Let $S$ be the set of all vertices of degree at least~3 in the block
$G_X$.  Note that by Theorem~\ref{t:extremal}, $S$ contains $a$ and
$b$, and since $G_X$ is 2-sparse, $S$ is a stable set of $G_X$, and
therefore of $G$. In what follows, we consider three cases: 
edge-colouring, total-colouring with $\Delta\geq 4$, and total-colouring with $\Delta=3$.

\subsection*{Edge-colouring}

By induction, we know that $G[Y \cup \{a, b\}]$ (this is not $G_Y$, we
do not use the marker vertex) is $\Delta$-edge-colourable.  Let $C_a$
and $C_b$ be the sets of the colours given to the edges incident to $a$
and $b$ respectively in such a colouring.  
We show how the elements of $G[Y \cup \{a,b\}]$ coloured so far give
the required list conditions in Lemma~\ref{t:list-edge-colouring-2-sparse}
in order to edge-colour the 2-sparse $G[X \cup \{a,b\}]$.
To do so, we associate to each edge incident
to $a$ and $b$ the list $\{1, \dots, \Delta\} \setminus
C_a$ and $\{1, \dots, \Delta\} \setminus C_b$ of colours respectively.  To the
other edges, we associate the list $\{1, \dots, \Delta\}$.  We apply
Lemma~\ref{t:list-edge-colouring-2-sparse} to $G[X \cup \{a, b\}]$ and
$S$ to colour the edges of $G[X \cup \{a, b\}]$.  We obtain an
edge-colouring of $G$ with colours $\{1, \dots, \Delta\}$.

\subsection*{Total-colouring, $\Delta\geq 4$}

By induction, we know that $G[Y \cup \{a, b\}]$ is
$(\Delta+1)$-total-colourable.  Let $C_a$ (resp.\ $C_b$) be the sets of the
colours given to $a$ (resp.\ $b$) and to the edges incident to $a$
(resp.\ $b$).  
We show how the elements of $G[Y \cup \{a,b\}]$ coloured so far give
the required list conditions in Lemma~\ref{l:extSparseDelta4}
in order to total-colour the 2-sparse $G[X \cup \{a, b\}]$.
To do so, we colour $a, b$ with the colours they have in
$G[Y \cup \{a, b\}]$ and all other vertices from $S$ with Colour~1.
We associate to each edge incident to $a$ and $b$ the list $\{1, \dots, \Delta+1\} \setminus C_a$ of colours
and $\{1, \dots, \Delta+1\}
\setminus C_b$ respectively.  To all the other edges incident to a
vertex $u$ of degree at least 3 (so, $u\not\in\{a,b\}$), we associate
the list $\{2, \dots, \Delta+1\}$.  By Lemma~\ref{l:extSparseDelta4}
applied to $G[X \cup \{a, b\}]$ and $S$, we extend this to a
total-colouring of $G$ with colours in $\{1, \dots, \Delta+1\}$.

\subsection*{Total-colouring, $\Delta=3$}

By Theorem~\ref{t:extremal}, vertices $a, b$ both have at least two
neighbours in $X$.  Hence, since $G$ has maximum degree 3, vertices $a,
b$ both have a unique neighbour in $Y$, say $a', b'$ respectively.
Notice that vertices $a'$ and $b'$ are distinct, for otherwise $Y=\{a'\}=\{b'\}$ and $G[Y\cup\{a,b\}]$ is a path.
Let $G_Y'$ be the block obtained from $G[Y \cup \{a, b\}]$ by
contracting $a$ and $b$ into a vertex $w_{ab}$.  Block $G_X$ is built
as previously defined by adding to $G[X \cup \{a, b\}]$ a vertex $w$
adjacent to $a, b$.

It is easy to check that $G_Y$ is chordless.  By induction we
total-colour $G_Y$ with 4 colours.  W.l.o.g.\ the vertex $w_{ab}$
receives Colour 1, edge $w_{ab}a'$ Colour 2, edge $w_{ab}b'$ Colour 3.
In $G_X$, precolour $a$ and $b$ with Colour 1, $wa$ with Colour 2, and
$wb$ with Colour 3.  Apply Lemma~\ref{l:extSparse} to $G_X$.  This
gives a total-colouring of $G_X$ with 4 colours.  A 4-total-colouring of
$G$ is obtained as follows: elements of $G$ that are in one of $G_X,
G_Y$ receive the colour they have in this block.  Edges $aa'$ and $bb'$
receive Colours 2 and 3 respectively.

\subsection*{Algorithm and complexity analysis}

Here we sketch the description of an algorithm that computes the
colourings whose existence is proved above. As usual, $n$ denotes the
number of vertices of the input graph, and $m$ the number of its
edges.  All the proofs in Section~\ref{s:2-sparse} are clearly
algorithmic and most arguments rely on very simple searches of the graph.
The only exception, which is the slowest step, is when we use
Theorem~\ref{th:konig}.  The edge-colouring whose existence is claimed
by Theorem~\ref{th:konig} can be obtained in time $O(nm)$, see
Section~20.1 in~\cite{schrijver:opticombA}.  Finding a proper 2-cutset
satisfying the conditions of Theorem~\ref{t:extremal} can be performed in time~$O(n^2m)$ as follows.
First observe that we may assume $G$ is biconnected, for otherwise we can use
Hopcroft and Tarjan~\cite{hopcroft.tarjan:447} algorithm to compute its biconnected components,
and reconstructing the colouring from the blocks is easy: just identify the colours of the vertices that are 1-cusets.
Now we have a biconnected graph. For each pair of vertices, (first) test whether such pair forms a proper 2-cutset 
and (second) choose, among all pairs, a proper 2-cutset 
with a block of minimum order. The time complexity of the first step
(checking whether a pair $\{a,b\}$ of vertices is a proper 2-cutset can be done in linear time $O(n+m)$, which is the time for characterizing the connected components of $G\setminus\{a,b\}$. Since there are $O(n^2)$ pairs of vertices to be tested, we have shown the
claimed complexity time of $O(n^2m)$ to find a proper 2-cutset
satisfying the conditions of Theorem~\ref{t:extremal}.

Now, a colouring can be obtained as follows. 
Find a proper 2-cutset as in Theorem~\ref{t:extremal}, and if one is found build the 
two blocks with respect to it, one of which is 2-sparse. The not 2-sparse block (if
any) is coloured recursively, and the colouring is extended to the 2-sparse block 
following the strategy in the proofs in Section~\ref{s:2-sparse}.  
Since a block has fewer vertices than the original graph, there are at most $n$ recursive calls, so,
the global running time is $O(n^3m)$.
\end{proof}

\section*{Acknowledgement}

We are grateful to Pierre Aboulker and Marko Radovanovi\'c for pointing out a
mistake in a preliminary version of this work.




\section*{References}
\begin{thebibliography}{9}

\bibitem{aboulkerRTV:propeller}
P.~Aboulker, M.~Radovanovi{\'c}, N.~Trotignon, K.~Vu{\v s}kovi{\'c}.
\newblock \emph{Graphs that do not contain a cycle with a node that has at least two
  neighbors on it}. SIAM J. Discrete Math. \textbf{26} (2012) 1510--1531.



\bibitem{Barbosa}
M. M. Barbosa, C. P. de Mello, J. Meidanis. \emph{Local conditions for edge-colouring of cographs}. Congr. Numer. \textbf{133} (1998) 45--55.


\bibitem{BehzadTese}
M. Behzad, \emph{Graphs and their chromatic numbers.}
Doctoral Thesis (1965) Michigan State University.


\bibitem{Bodlaender}
H. L. Bodlaender. \emph{Polynomial algorithms for graph isomorphism and chromatic index on partial k-trees}. J. Algorithms \textbf{11} (1990) 631--643.


\bibitem{Bojarshinov}
V. A. Bojarshinov. \emph{Edge and total colouring of interval graphs}. Discrete Appl. Math. \textbf{114} (2001) 23--28.


\bibitem{Borodin}
O. V. Borodin, A. V. Kostochka, D. R. Woodall.
\emph{List edge and list total colourings of multigraphs}.
J. Combin. Theory Ser. B \textbf{71} (1997) 184--204.


\bibitem{CaiEllis}
L. Cai, J. A. Ellis. \emph{NP-Completeness of edge-colouring some restricted graphs}. Discrete Appl. Math. \textbf{30} (1991) 15--27.
%
%


\bibitem{ChristianeBipartidos}
C. N. Campos, C. P. de Mello. \emph{The total chromatic number of some bipartite graphs.} Ars Combin. \textbf{88} (2008) 335--347.

\bibitem{ChristianeSplitIndifference}
C. N. Campos, C. M. H. de Figueiredo, R. C. S. Machado, and C. P. de Mello. \emph{The total chromatic number of split-indifference graphs.}  Discrete Math. \textbf{312} (2012) 2690--2693.


\bibitem{ChenFuKo}
B.-L. Chen, H.-L. Fu, M. T. Ko. \emph{Total chromatic number and chromatic index of split graphs}. J. Combin. Math. Combin. Comput. \textbf{17} (1995) 137--146.

\bibitem{Dually}
C. M. de Figueiredo, J. Meidanis, C. P. de Mello. \emph{Total-chromatic number and chromatic index of dually chordal graphs.} Inform. Process. Lett. \textbf{70} (1999) 147--152.

\bibitem{Celina}
C. M. H. de Figueiredo, J. Meidanis, C. P. de Mello, C. Ortiz. \emph{Decompositions for the edge colouring of reduced indifference graphs}. Theoret. Comput. Sci. \textbf{297} (2003) 145--155.


\bibitem{join1}
C. De Simone, C. P. de Mello. \emph{Edge-colouring of join graphs}.  Theoret. Comput. Sci. \textbf{355} (2006) 364--370.


\bibitem{dirac:chordless}
G. A. Dirac.
\newblock Minimally 2-connected graphs.
\newblock {\em J. Reine Angew. Math.},
  \textbf{228} (1967) 204--216.
%
%

\bibitem{SP-choose-Delta4}
T. J. Hetherington, D. R. Woodall. \emph{Edge and total choosability of near-outerplanar graphs.} Electron. J. Combin. \textbf{13} (2006) \#R98.

\bibitem{Hilton}
A. J. W. Hilton, J. Liu, C. Zhao. \emph{The total chromatic numbers of joins of 2-sparse graphs.} Australas. J. Combin. \textbf{28} (2003) 93--105.

\bibitem{Holyer}
I. Holyer. \emph{The NP-completeness of edge-colouring}.  SIAM J. Comput. \textbf{10} (1982) 718--720.

\bibitem{hopcroft.tarjan:447}
J. E. Hopcroft and R. E. Tarjan.
\newblock Algorithm 447: efficient algorithms for graph manipulation.
\newblock {\em Communications of the ACM} \textbf{16} (1973) 372--378.


\bibitem{Johnson}
D. S. Johnson. \emph{The NP-completeness column: an ongoing guide}. J. Algorithms \textbf{6} (1985) 434--451.
%
%
%
%
\bibitem{Leven}
D. Leven, Z. Galil. \emph{NP-completeness of finding the chromatic index of regular graphs}. J. Algorithms \textbf{4} (1983) 35--44.

\bibitem{ISK4}
B. L\'ev\^eque, F. Maffray, N. Trotignon. \emph{On graphs with no induced
  subdivision of  $K_4$.} J. Comb. Theory Ser. B \textbf{102} (2012) 924--947.

\bibitem{OneKindJoin}
G. Li, L. Zhang. \emph{Total chromatic number of one kind of join graphs.} Discrete Math. \textbf{306} (2006) 1895--1905.


\bibitem{JoinCobipartite}
R. C. S. Machado, C. M. H. de Figueiredo. \emph{Decompositions for edge-colouring join graphs and cobipartite graphs}. Discrete Appl. Math. \textbf{158} (2010) 1336--1342.

\bibitem{Edge-tv}
R. C. S. Machado, C. M. H. de Figueiredo, K. Vu{\v s}kovi{\'c}. \emph{Chromatic index of graphs with no cycle with unique chord}. Theoret. Comput. Sci. \textbf{411} (2010) 1221--1234.

\bibitem{Total-tv}
R. C. S. Machado and C. M. H. de Figueiredo. \emph{Total chromatic number of unichord-free graphs}. Discrete Appl. Math. \textbf{159} (2011) 1851--1864.


\bibitem{Total-tv-square}
R. C. S. Machado, C. M. H. de Figueiredo, Nicolas Trotignon. \emph{Complexity of colouring problems restricted to unichord-free and { square,unichord }-free graphs}. Discrete Appl. Math. \texttt{DOI:10.1016/j.dam.2012.02.016}.

%

\bibitem{SanchezArroyo}
C. J. H. McDiarmid, A. S{\'a}nchez-Arroyo. \emph{Total colouring regular bipartite graphs is NP-hard.} Discrete Math. \textbf{124} (1994) 155-162.

\bibitem{SplitIndifference}
C. Ortiz, N. Maculan, J. L. Szwarcfiter. \emph{Characterizing and edge-colouring split-indifference graphs}. Discrete Appl. Math. \textbf{82} (1998) 209--217.
%
%

\bibitem{plummer:68}
M. D. Plummer.
\newblock On minimal blocks.
\newblock {\em Trans. Amer. Math. Soc.} 
\textbf{134} (1968) 85--94.


\bibitem{SanchezArroyo1}
A. S{\'a}nchez-Arroyo. \emph{Determining the total colouring number is NP-hard.} Discrete Math. \textbf{78} (1989) 315--319.

%
\bibitem{Sanders}
D. P. Sanders, Y. Zhao. \emph{Planar graphs of maximum degree seven are class~I}. J. Combin. Theory Ser. B \textbf{83} (2001) 201--212.
%
%

\bibitem{schrijver:opticombA}
A.~Schrijver.
\newblock {\em Combinatorial Optimization, Polyhedra and Efficiency}, volume~A.
\newblock Springer, 2003.


\bibitem{tv}
N. Trotignon, K. Vu{\v s}kovi{\'c}.
\emph{A structure theorem for graphs with no cycle with a unique chord and
its consequences}. J. Graph Theory \textbf{63} (2010) 31--67.

\bibitem{Vizing}
V. G. Vizing. \emph{On an estimate of the chromatic class of a
$p$-graph} (in Russian). Diskret. Analiz \textbf{3} (1964) 25--30.



\bibitem{WeifanWang}
W. Wang. \emph{Total chromatic number of planar graphs with maximum degree ten.} J. Graph Theory \textbf{54} (2007) 91--102.

\bibitem{WangPang}
S.-D. Wang, S.-C. Pang. \emph{The determination of the total-chromatic number of series-parallel graphs.} Graphs Combin. \textbf{21} (2005) 531--540.

%
%

\bibitem{SP-choose-Delta3}
D. R. Woodall. \emph{Total 4-choosability of series-parallel graphs.} Electron. J. Combin. \textbf{13} (2006) \#R97.

\bibitem{Yap}
H.-P. Yap. \emph{Total colourings of graphs.} Lecture Notes in Mathematics. Springer-Verlag (1996) Germany.

\end{thebibliography}
\end{document}